\documentclass[conference]{IEEEtran}

\usepackage{amsthm}
\usepackage{amsfonts}
\usepackage{amssymb}
\usepackage{latexsym}
\usepackage{cite}
\usepackage[cmex10]{amsmath}
\usepackage{algorithmic}
\usepackage{url}
\usepackage[tight,footnotesize]{subfigure}

\usepackage[pdftex]{graphicx}

\usepackage{array}
\newcolumntype{L}[1]{>{\raggedright\let\newline\\\arraybackslash\hspace{0pt}}m{#1}}

\newcommand{\rrn}{r_r}
\newcommand{\erfc}{\mathrm{erfc}\,}
\newcommand{\rxi}[1]{\mathrm{Rx}#1}

\usepackage{amsthm}

\newtheorem{theorem}{Theorem}

\newtheorem{lemma}[theorem]{Lemma}

\begin{document}

\title{Detection Algorithms for Molecular MIMO}

\author{\IEEEauthorblockN{Bon~Hong~Koo, H. Birkan~Yilmaz, Chan-Byoung Chae}
\IEEEauthorblockA{School of Integrated Technology\\
Yonsei Institute of Convergence Technology\\
Yonsei University, Korea\\
Email: \{harpeng7675,birkan.yilmaz,cbchae\}@yonsei.ac.kr}
\and
\IEEEauthorblockN{Andrew Eckford}
\IEEEauthorblockA{Dept. of Electrical Engineering and Computer Science\\
York University, Toronto, Canada\\
Email: aeckford@yorku.ca}
}


\maketitle

\begin{abstract}
In this paper, we propose a novel design for molecular communication in which both the transmitter and the receiver have, in a 3-dimensional environment, multiple bulges (in RF communication this corresponds to antenna). The proposed system consists of a fluid medium, information molecules, a transmitter, and a receiver. We simulate the system with a one-shot signal to obtain the channel's finite impulse response. We then incorporate this result within our mathematical analysis to determine interference. Molecular communication has a great need for low complexity, hence, the receiver may have incomplete information regarding the system and the channel state. Thus, for the cases of limited information set at the receiver, we propose three detection algorithms, namely adaptive thresholding, practical zero forcing, and Genie-aided zero forcing.
\end{abstract}

\IEEEpeerreviewmaketitle

\section{Introduction} 
If operating at the nano-scale is to make an impact at the macro scale, there must be high cooperation among multiple devices~\cite{nakano2013molecularC,kuran2010energyMF}. Therefore, nanonetworking emerges as a new paradigm and molecular communication, as an interdisciplinary branch, draws the attention. The literature has proposed  various molecular communication systems, such as molecular communication via diffusion (MCvD), calcium signaling, microtubules, DNA micro-arrays, pheromone signaling, and bacterium-based communication~\cite{nakano2013molecularC, kuran2012calciumSO, kuran2010energyMF, lio2012opportunisticRT, far12NanoBio}. In an MCvD, a number of micro- and nano-machines reside in a viscous environment and communicate through molecules that are emitted into the medium. Following the physical characteristics of the diffusion channel, these molecules propagate through the environment. Some of these molecules arrive at the receiver (i.e., hit the receiver) and form chemical bonds with the receptors on the surface of the receiver. The properties of these received molecules (e.g., concentration and/or type) constitute the received signal~\cite{kuran2010energyMF,kim2013novelMT}. 

To increase MCvD performance, many enhancements are proposed in the literature such as incorporating inter symbol interference (ISI) mitigation techniques~\cite{yilmaz2014simulationSO_SIMPAT, NKim_TNB}, using multiple molecule types~\cite{kuran2012interferenceEO}, using multiple input multiple output (MIMO) techniques~\cite{meng2012mimoCB}. The authors in~\cite{meng2012mimoCB}, however, mainly focused on multiuser interference and paid scant attention to the ISI. In molecular communication, ISI is the main source of communication impairment and must be analyzed precisely. Therefore, this paper focuses on introducing and enhancing MIMO analysis.

We propose a molecular MIMO system and, considering the channel and interference models, introduce new detection algorithms. The interference in the system is caused by the ISI and inter link interference (ILI). This is where ISI originates from the previous emissions of the corresponding antenna and ILI arises from the other antenna. To the best of our knowledge, this is the first work that considers MIMO for molecular communication while taking into account the ISI and ILI. We model the channel's impulse response by modifying the single input single output (SISO) channel model in a 3-dimensional (3-D) environment~\cite{yilmaz2014_3dChannelCF}. Consequently, we analyze the system performance via the MIMO simulator developed through MATLAB.

The rest of this paper is organized as follows. In Section~\ref{model}, we describe the system model including topology, propagation, and communication models. In Section~\ref{fit_n_algorithm}, we detail the channel estimation method and the proposed detection algorithms. In Section~\ref{results}, the results and discussions are presented. Finally, the conclusion is given in Section~\ref{conclusion}.

\section {System Model}
\label{model}
Consider a molecular communication system in a 3-D environment with two point sources and two spherical receiver antennas.\footnote{Throughout this paper, we use the terms bulge and antenna interchangeably.} The transmitter releases, without any directional preferences, a certain number of messenger molecules at once. Without colliding into one another or undergoing any chemical reactions, the emitted molecules travel along the fluid medium via diffusion. When a molecule hits the boundary of a spherical antenna, it is immediately absorbed by the receiving antenna and removed from the medium. We assume that the transmitter-receiver pair is synchronized and the receiver can count the number of received molecules during a symbol duration.

During communication, the transmit antennas convey independent messages to their corresponding receiver bulges. Each link of transceivers uses the same type of molecule and the molecules from the other transmitter cause ILI. 

\subsection{Topology and Propagation Model}
\label{topology}	
\begin{figure}[t]
\centering
\includegraphics[width=0.90\columnwidth,keepaspectratio]
{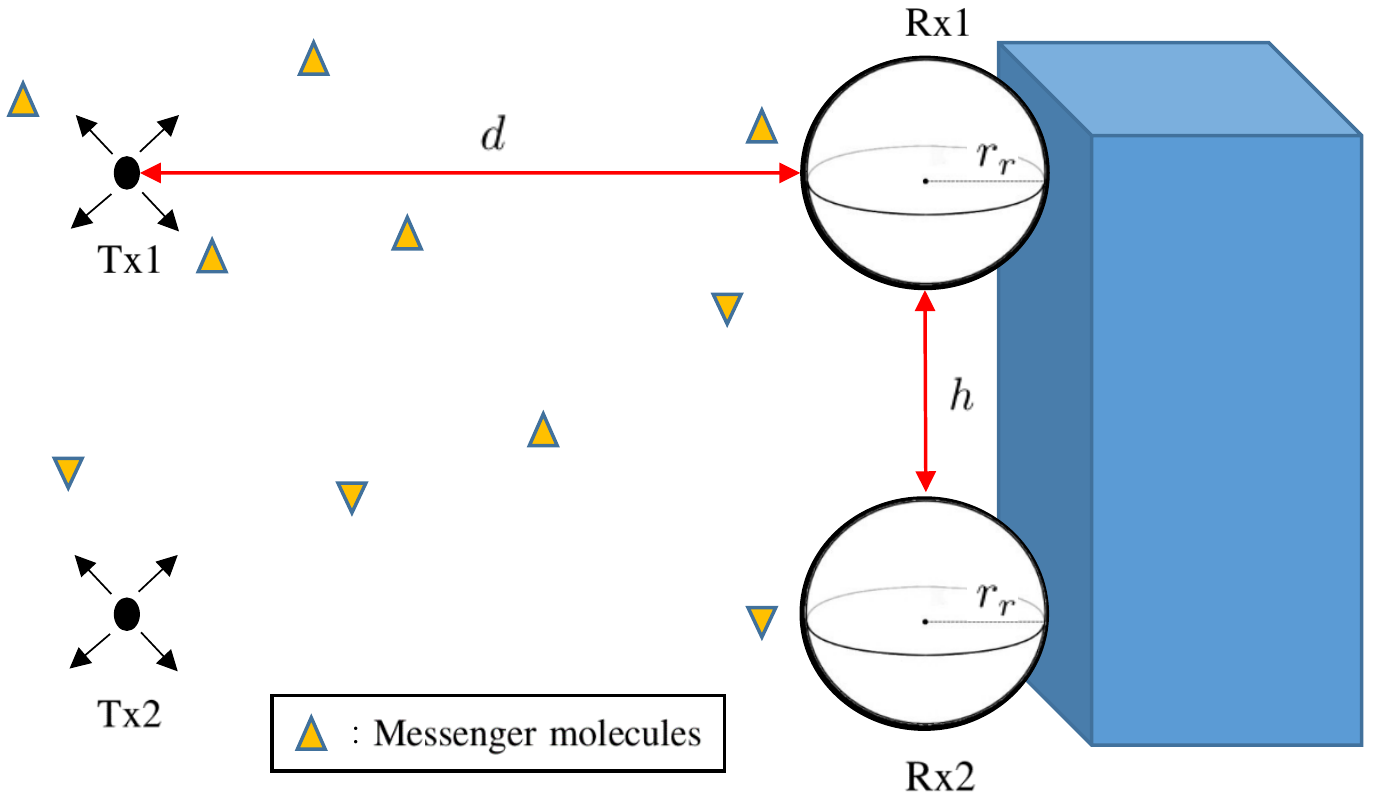}
\caption{Topological model of molecular $2\times2$ MIMO system.}
\label{Fig:Topology_model}
\end{figure}

As shown in Fig.~\ref{Fig:Topology_model}, there are two transmitter antennas, Tx1 and Tx2,  placed $d$ distance apart from the corresponding receiver bulges, Rx1 and Rx2. We have two receiver bulges with the same radius $\rrn$, which are placed $h$ distance apart. The centers of the receiver bulges, Tx1,  and Tx2 lie on the corners of a rectangular grid. The receiver bulges are attached to the receiver body and we assume that just the antennas are capable of receiving molecules, which is a realistic assumption considering many examples found in the nature. For example, epithelial cells, neurons, and migrating cells are examples of polarized cells that have heterogeneous receptor
deployments, which is an adaptation to the environment and the signaling mechanism. 

When a molecule is released from a point source, its movement in a fluid is governed by diffusion and drift. Drift is applicable if there is a flow and we leave the drift case to future work. The dynamics of diffusion can be described by Brownian motion. Derived in~\cite{yilmaz2014_3dChannelCF} is the fraction of molecules that are absorbed by a single spherical receiver antenna in a 3-D environment until time $t$ with given parameters
\begin{equation}
\label{eqn_first_hitting_3d}
F(t|\rrn,\, d,\, D)= \frac{\rrn}{\rrn+d}\, \erfc\left( \frac{d}{\sqrt{4Dt}} \right)
\end{equation}
where $D$ denotes the diffusion coefficient. In our setup, however, there are two receiver bulges and we cannot directly use (\ref{eqn_first_hitting_3d}). Therefore, we simulate the Brownian motion for released particles within the given MIMO setup by using 
\begin{align}
\begin{split}
(x_{t}, y_{t}, z_{t}) &= (x_{t-\Delta t}, y_{t-\Delta t}, z_{t-\Delta t}) + (\Delta x, \Delta y, \Delta z)\\
\Delta x & \sim \mathcal{N}(0, 2D\Delta t) \\
\Delta y & \sim \mathcal{N}(0, 2D\Delta t) \\
\Delta z & \sim \mathcal{N}(0, 2D\Delta t) 
\end{split}
\label{eqn_propagation_model}
\end{align}
where $x_{t}$, $y_{t}$, $z_{t}$, and $\mathcal{N}(\mu, \sigma^2)$ are the particles' positions at each dimension at time $t$, and the normal distribution with mean $\mu$ and variance $\sigma^2$. The Brownian motion simulator for the MIMO setup is a modified version of the  simulator that is developed for a SISO case in~\cite{yilmaz2014simulationSO_SIMPAT}. During its trip, if a molecule hits one of the spherical receiver bulges, then it is absorbed and removed from the environment. Therefore, a molecule can contribute to the signal just once and we have four different $F(t|r_r, d, D)$ values depending on the molecule's emission source and hitting bulge. We use the simulation data to formulate $F(t|r_r, d, D)$ for the 2$\times$2 molecular MIMO setup and utilize it for the analysis.

 \subsection{Communication Model}
\label{communication}	

To encode information, we use the binary concentration shift keying (BCSK) modulation technique. We denote $Q_1$ and $Q_0$ as the number of molecules released to send bit-1 and ${\mbox{bit-0}}$, respectively. The transmitter has independent sets of bits $x_1$ and $x_2$ for their own messages. During the $n^{\mathrm{th}}$ symbol, Tx1 and Tx2 send $x_1[n]$ and $x_2[n]$ each by releasing ${Q}_{{x}_{1}[n]}$ and ${Q}_{{x}_{2}[n]}$ molecules at the start of the symbol time, and wait until the next emission time. The duration between consecutive symbols is called the symbol duration and is denoted by ${t}_{s}$. The number of molecules absorbed at the receiver stochastically follows a binomial distribution with a hitting probability, which is related to $t_s$, $d$, $r_r$, and $D$~\cite{kuran2010energyMF}. We define $F_{ij}(t_1, t_2)$ as the probability of hitting to Rx$i$ between $t_1$ and $t_2$ for a molecule that is released from a Tx$j$. So we can define the random variable $\mathcal{S}_{ij}(t_1,t_2)$ as follows:
\begin{equation}
\label{eq_binom_rv_sij}
\mathcal{S}_{ij}(t_1,t_2) \triangleq \mathcal{B}\left(1,F_{ij}(t_1,t_2)\right),
\end{equation}
where $\mathcal{B}(n,p)$ denotes the binomial random variable with $n$ trials and success probability $p$. Binomial random variable in \eqref{eq_binom_rv_sij} can be considered as Bernoulli trial.   $\mathcal{S}_{ij}$ is utilized while evaluating the number of received molecules at Rx$i$ that originates from Tx$j$.

In this paper, we consider two types of interference sources for the receiving bulge: interference from the previous symbols of the  corresponding transmitter (i.e., ISI) and interference from the current and the previous symbols of the other link (i.e., ILI). The ISI at the $n^{\mathrm{th}}$ symbol can be modeled as a sum of interference due to the previous symbols. The ILI at the $n^{\mathrm{th}}$ symbol can be modeled as a sum of interference due to the other link emissions including the current symbol.  Hence, the interference model at Rx$i$ can be expressed as 
\begin{equation}
\label{interferences}
\begin{split}
{I}_{i}[n] &={ISI}_i[n]+{ILI}_i[n]\\
{ISI}_i[n] &=\sum\limits_{k=1}^{n-1}Q_{x_{i}[n-k]}{\mathcal{S}}_{ii}[k]\\
 {ILI}_1[n] &=\sum\limits_{k=0}^{n-1}Q_{x_{2}[n-k]}{\mathcal{S}}_{12}[k]\\
 {ILI}_2[n] &=\sum\limits_{k=0}^{n-1}Q_{x_{1}[n-k]}{\mathcal{S}}_{21}[k]\\
 {\mathcal{S}}_{ij}[k] &\triangleq{\mathcal{S}}_{ij}(kt_{s},(k\!+\!1)t_{s})
\end{split}
\end{equation}
where $I_{i}[n]$, $ISI_{i}[n]$, and $ILI_{i}[n]$ denote random variables of total interference, ISI, and ILI induced at Rx$i$ at the $n^\mathrm{th}$ time slot respectively. Note that the summation in the ILI term starts from zero, because the current symbol of the other link also induces interference.
The channel output at Rx$i$ for the $n^\mathrm{th}$ time slot can be written as 
\begin{equation}
\label{eq_channel_output}
y_{\mathrm{Rx}i}[n]=((Q_1\!-\!Q_0)x_i[n]\!+\!Q_0)\mathcal{S}_{ii}[0]\!+\!{I}_{i}[n]\!+\!{n}_{i}[n]
\end{equation}
where, for each time slot, $y_{\mathrm{Rx}i}$ is the random variable of the number of received molecules and ${n}_{i}$ denotes the molecular noise that can occur due to outer molecular invasion or decomposition at Rx$i$. The effect of noise is assumed to be a normal distribution $\mathcal{N}(\mu_{n},\sigma_{n}^{2})$. In this paper, the transmitter emits zero molecules to send bit-$0$ to reduce the energy consumption and separate the signal amplitudes as much as possible. Therefore, substituting $Q_0$ in \eqref{eq_channel_output} with zero and writing in matrix form for each symbol slot yields the following
\begin{align}
\label{eq_y_inMatrixForm}
\begin{split}
\begin{bmatrix} y_{\mathrm{Rx}1} \\ y_{\mathrm{Rx}2}\end{bmatrix} \!=\! \begin{bmatrix} Q_1\mathcal{S}_{11}[0] & 0\\ 0& Q_1\mathcal{S}_{22}[0]\end{bmatrix} \begin{bmatrix} x_1 \\ x_2\end{bmatrix} \!+\! \begin{bmatrix} {I}_{1} \\ {I}_{2}\end{bmatrix} \!+\! \begin{bmatrix} {n}_{1} \\ {n}_{2}\end{bmatrix}
\end{split}
\end{align}
for $2\times{2}$ MIMO system. We write \eqref{eq_y_inMatrixForm} in short as
\begin{equation}
\pmb{y} =\pmb{H} \pmb{x}+\pmb{I}+\pmb{n}.
\end{equation} 
The details of $\pmb{I}$ will be given in the following section.

\section{Fitting Channel Parameters and Proposed Detection Algorithms}
\label{fit_n_algorithm}
We need to have a formula for the first hitting probability in molecular MIMO setup to formulate ${F}_{ij}(t)$. Therefore, extensive simulations are carried out to understand the underlying formula and we use nonlinear curve fitting on the simulation data. After having the fitted cumulative distribution functions (CDFs) we carry out our analytical derivations via utilizing approximate ${F}_{ij}(t)$ functions. We also introduce the proposed detection algorithms and optimal thresholds in this section.

\subsection{Fitting Channel Parameters}
\label{fitting}
We use a model function that is coherent with  (\ref{eqn_first_hitting_3d}) to fit the simulation data (i.e., the formula is coherent with molecular SISO system in a 3-D environment with some control parameters). The model function structure for non-linear fitting is as follows:
\begin{equation}
F_{\text{model}}(t|\rrn,\, d,\, D) =  \frac{b_1\,\rrn}{d\!+\!\rrn}  \erfc\left( \frac{d}{{(4D)^{b_2} \,t^{b_3}}} \right)
\label{eqn_fpp_model_function}
\end{equation}
where $b_1$, $b_2$, and $b_3$ are controllable parameters. We run extensive simulations for different parameter sets and estimate mean CDF of the hitting molecules. The hitting molecules are separated according to where they are originated for finding ${F}_{11}$, ${F}_{12}$, ${F}_{21}$, and ${F}_{22}$. Due to the symmetry of the topology, ${F}_{11}$ and ${F}_{12}$ are very close to ${F}_{22}$ and ${F}_{21}$ in the simulation data.  
\begin{figure*}[t]
	\begin{center}
		\subfigure[Fitted values of $b_1$]
		{\includegraphics[width=0.32\textwidth,keepaspectratio]
			{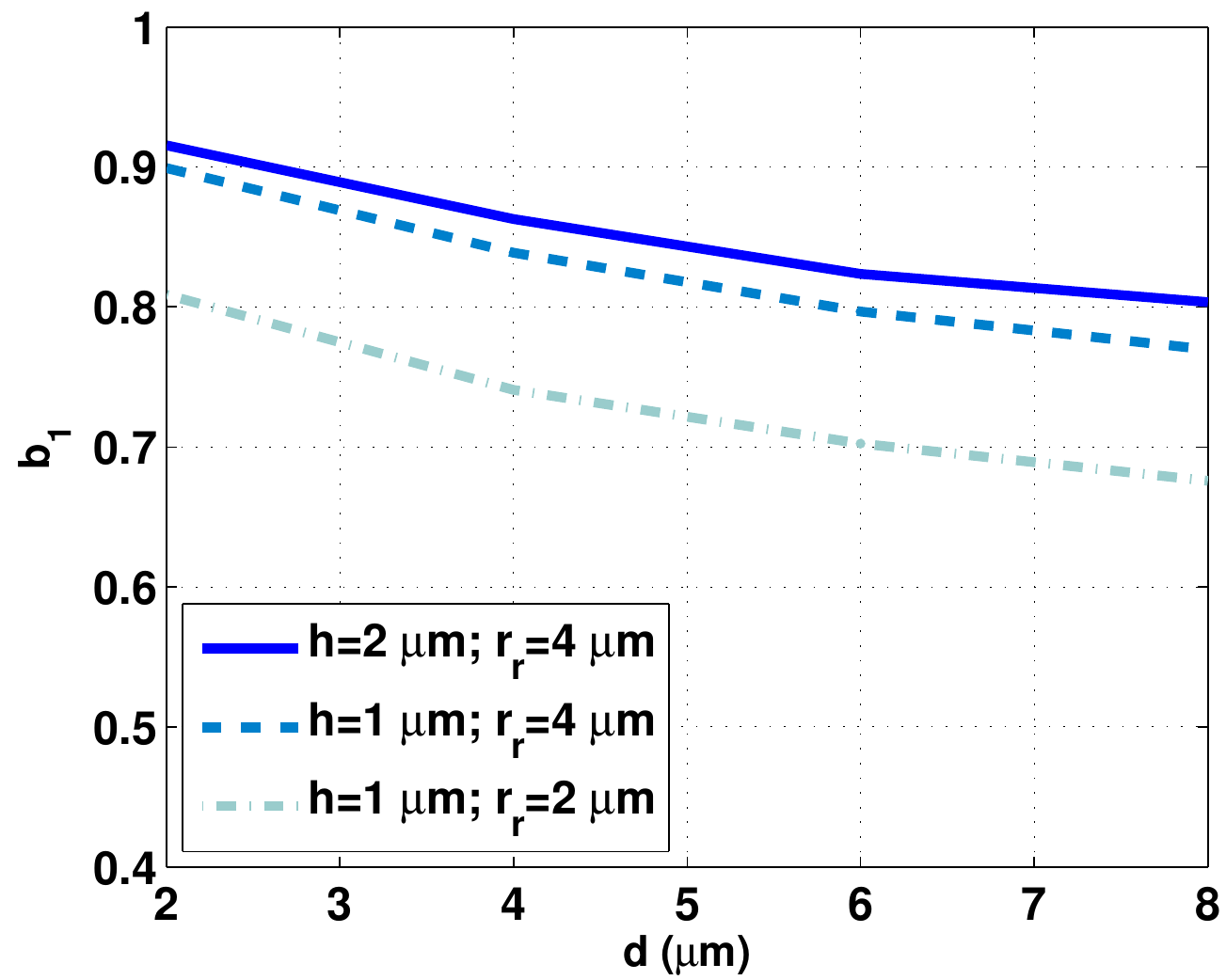}
			\label{subfig_fitting_b1} } 
		\subfigure[Fitted values of $b_2$]
		{\includegraphics[width=0.32\textwidth,keepaspectratio]
			{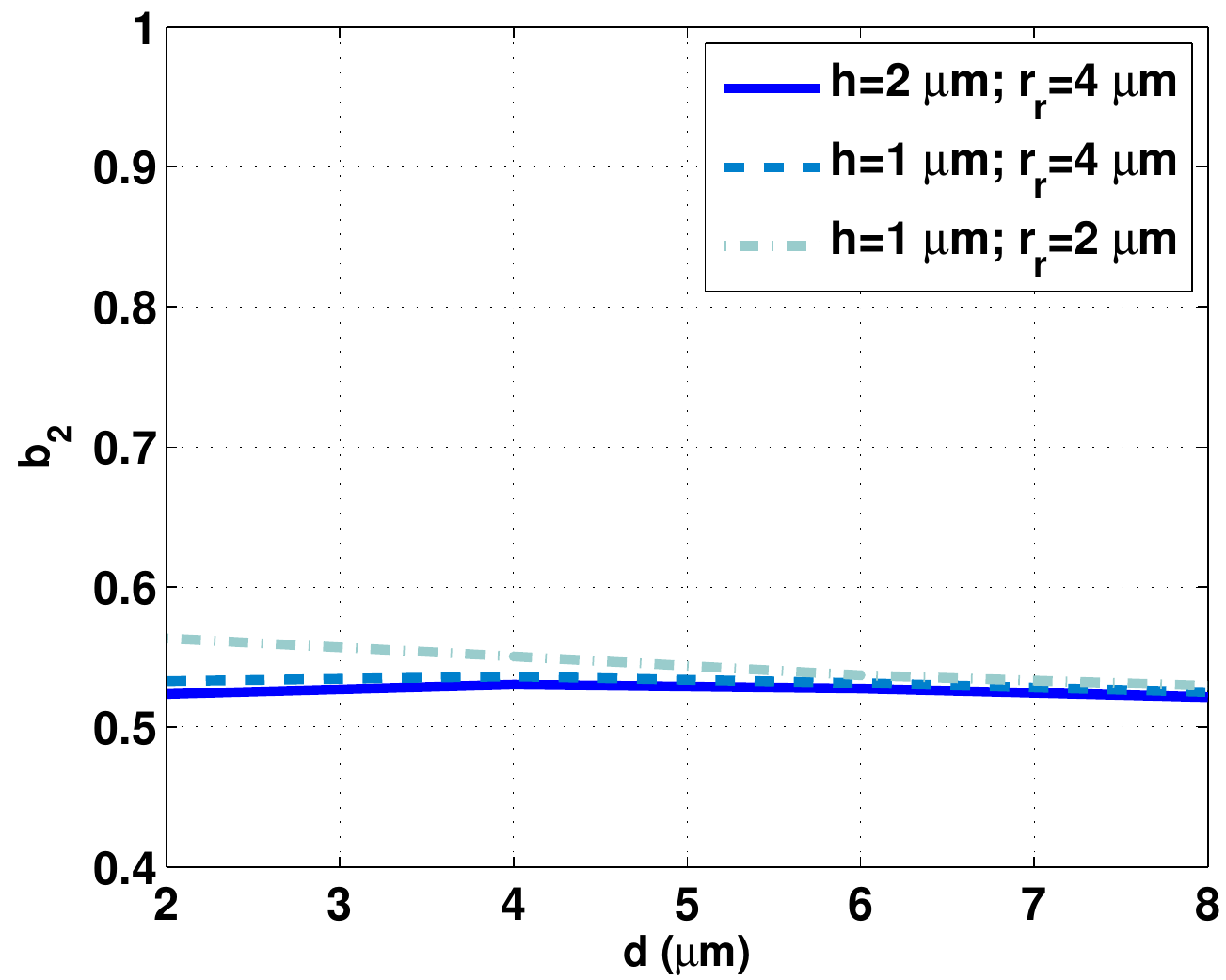}
			\label{subfig_fitting_b2} }
		\subfigure[Fitted values of $b_3$]
		{\includegraphics[width=0.32\textwidth,keepaspectratio]
			{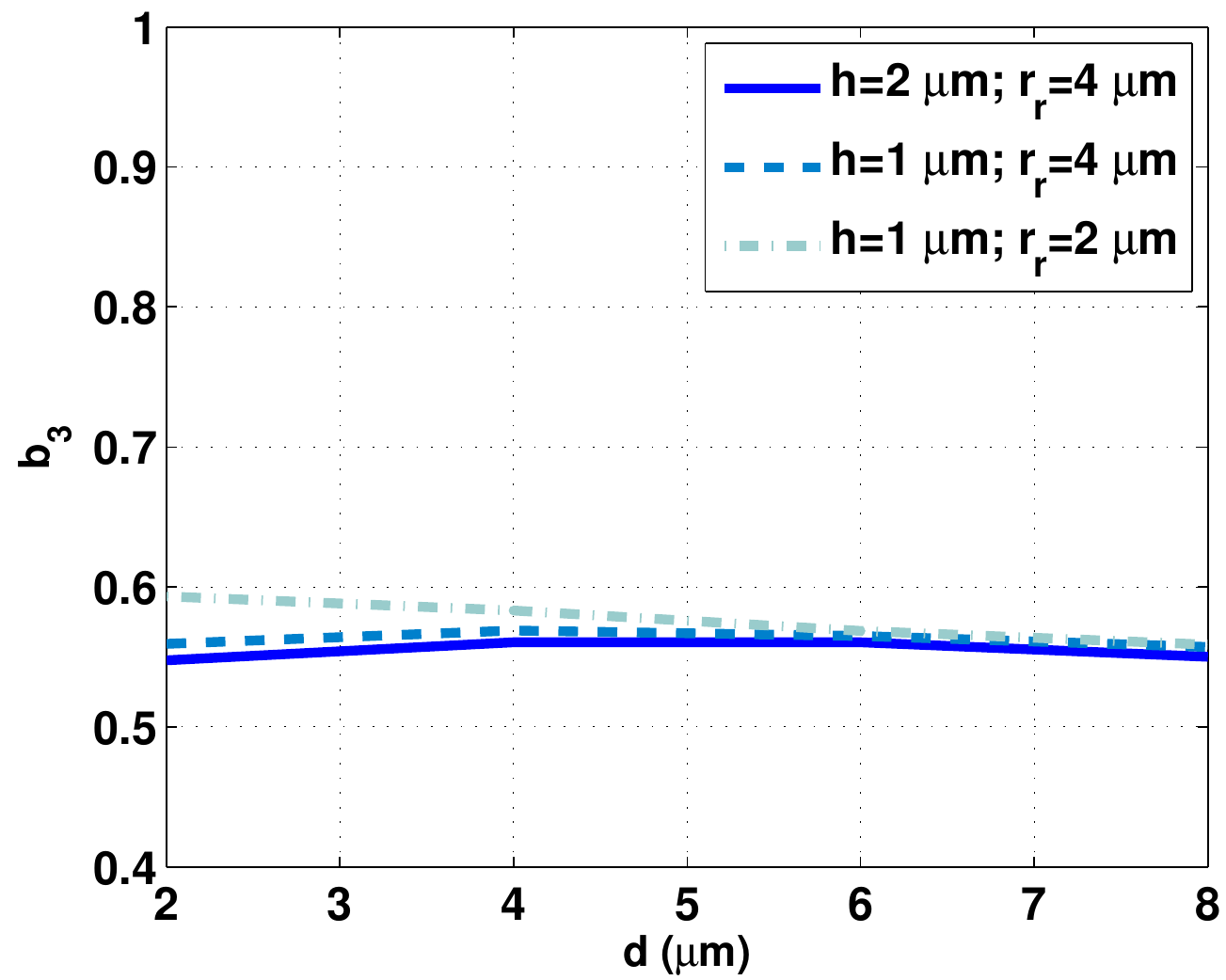}
			\label{subfig_fitting_b3} }
	\end{center}
	\caption{Fitted model parameters of ${F}_{11}(t)$ for different $h$ and $\rrn$ values ($D=50 \mu m^2/s$).}
	\label{fig_fitting_all}
\end{figure*}

NonLinearModel class in MATLAB is used for fitting nonlinear regression models. We implement the model function and utilize the simulation outputs to have a closed form CDF estimation obeying (\ref{eqn_fpp_model_function}). Fig.~\ref{fig_fitting_all} depicts the distance versus fitted model parameters for different $h$ and $\rrn$ values. The values of $b_2$ and $b_3$ change little when the distance is increased and, similar to the SISO case, they are close to 0.55.

\subsection{Detection Algorithms}
\label{algorithm}
\begin{figure}[t]
\centering
\includegraphics[width=0.99\columnwidth,keepaspectratio]
{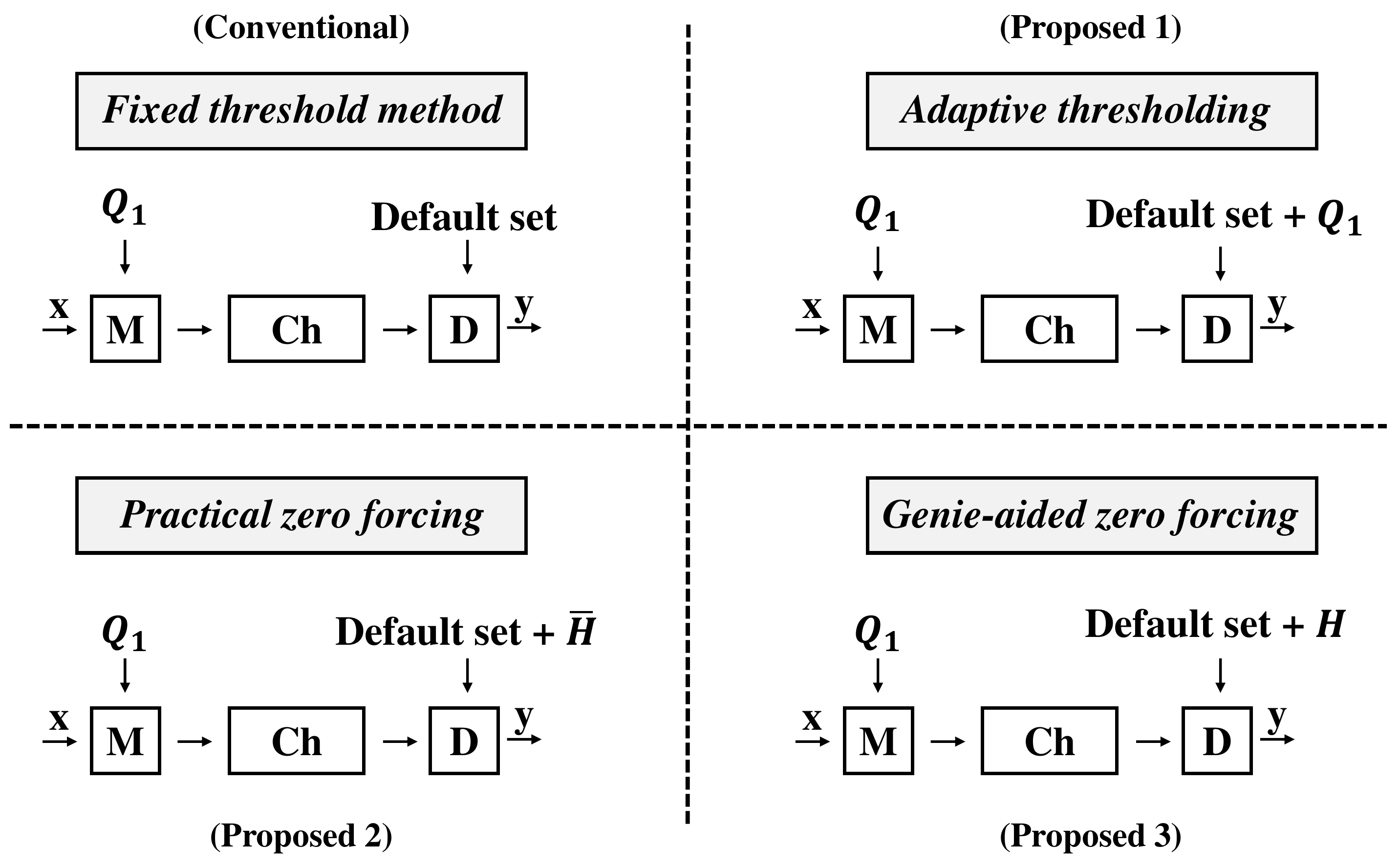}
\caption{Representation of detection algorithms and their requirements.}
\label{Fig:Algorithm}
\end{figure}
In this section, we introduce four detection algorithms. Each of them requires a different set of information given to the receiver. Fig.~\ref{Fig:Algorithm} describes the algorithms in terms of the information required. There is a default set that is needed commonly for all the algorithms. The default set consists of system parameters $D$, $t_s$, and topology parameters such as $d$, $h$, and $r_r$. First algorithm works with only the default set and we name it the \textit{fixed threshold method}. It uses a predetermined threshold and does not adapt to varying $t_s$ or $Q_1$. The second algorithm, called \textit{adaptive thresholding}, additionally requires $Q_1$, which is not a big assumption since $Q_1$ is determined with the communication protocol and the modulation. The output of the detector, $\hat{\pmb{y}}_{a}$, is formulated in ({\ref{eqn_adaptive}}) and the algorithm calculates the optimal decision threshold accordingly.
\begin{equation}
\label{eqn_adaptive}
\hat{\pmb{y}}_{a}=\frac{1}{Q_1}\pmb{y}=\frac{1}{Q_1}({\pmb{H}\pmb{x}}+{\pmb{I}}+\pmb{n})
\end{equation}
The third algorithm, called \textit{practical zero forcing} method, needs the average channel response matrix, which is denoted by $\bar{\pmb{H}}$. Inspired from the zero forcing of the conventional communication strategy, the formulation is given in ({\ref{eqn_practical}}). The output of the detector, $\hat{\pmb{y}}_{p}$, is formulated as
\begin{align}
\label{eqn_practical}
\hat{\pmb{y}}_{p}&={\bar{\pmb{H}}}^{-1}{\pmb{H}\pmb{x}}+{\bar{\pmb{H}}}^{-1}{\pmb{I}}+{\bar{\pmb{H}}}^{-1}\pmb{n}.
\end{align}
Lastly, the receiver is aided by Genie so that the exact channel states are known for every signal reception time. We call the algorithm for this case \textit{Genie-aided zero forcing}. The output of the detector, $\hat{\pmb{y}}_{g}$, is formulated in ({\ref{eqn_genie}}). It provides the best performance but is not feasible since the randomness of the molecular communication channel originates from the Brownian motion of molecules and is hard to acquire instantaneously.
\begin{align}
\label{eqn_genie}
\hat{\pmb{y}}_{g}&=\pmb{x}+{\pmb{H}}^{-1}{\pmb{I}}+{\pmb{H}}^{-1}\pmb{n}.
\end{align}

Our main contribution includes proving that the \textit{adaptive thresholding} and the \textit{practical zero forcing} methods perform exactly the same for the symmetrical MIMO topology. 

\begin{theorem}
	\label{theorem_adap&prac}
	When the centers of the transmitter and the receiver antennas form a rectangular grid, the detector outputs of the \textit{adaptive thresholding} and the \textit{practical zero forcing} methods satisfy $\,\hat{\pmb{y}}_{a}=A_{0}\,\hat{\pmb{y}}_{p}$ ($A_0$ denotes the hitting probabilities of a molecule to the intended bulge at the current symbol duration).
\end{theorem}
\begin{proof}
	The topological symmetry guarantees the random variables $\mathcal{S}_{11}[k]$ and $\mathcal{S}_{22}[k]$ have equal statistical parameters for a positive integer $k$. They are both  binomial random values with success probability of $A_k$. The transmitter sends $Q_1$ molecules for transmitting a bit-1. Therefore, the diagonal entries of $\pmb{H}$ follow a binomial distribution and approximated to the normal distribution as follows:
	\begin{equation}
	\label{eqn_b2n}
	\pmb{H}(i,i) \sim \mathcal{B}(Q_{1},A_0) \approx \mathcal{N}\left(Q_{1}{A}_{0},Q_{1}A_{0}(1-A_{0})\right)
	\end{equation}
	and the channel mean $\bar{\pmb{H}}$ equals $Q_{1}A_{0}\pmb{E}$ where $\pmb{E}$ denotes a $2\times2$ identity matrix. It leads ${\bar{\pmb{H}}}^{-1}=(1/Q_{1}A_{0})\pmb{E}$ and $\hat{\pmb{y}}_{a}$ in ({\ref{eqn_adaptive}}) becomes a multiplication of $A_{0}$ and $\hat{\pmb{y}}_{p}$ in ({\ref{eqn_practical}}).
\end{proof}
Theorem~\ref{theorem_adap&prac} ensures that both methods on average perform the same, since the signal detection properties (i.e., the detector outputs) are similar up to a constant multiple.

\subsection{Interference Formulations and Optimal Thresholds}
\label{threshold}
In this section, we formulate the interference and find the optimal  decision rule for the \textit{adaptive thresholding} and \textit{practical zero forcing} methods analytically. To derive the optimal decision threshold, the study uses the maximum-a-posterior (MAP) method.

As the receiver is unaware of the parameter $Q_1$ for the case of \textit{fixed threshold method}, we have to predetermine a range of $Q_1$ to use and decide the decision threshold, $\eta_{f}$, to minimize the bit error rate (BER) for all $Q_1$. Note that the analysis of \textit{Genie-aided zero forcing} inherits difficulties of acquiring instantaneous $\pmb{H}$. Hence, the optimal thresholds for \textit{Genie-aided zero forcing} are found empirically.

We should consider the interference when determining the thresholds and consequently the received symbol. Topological symmetry ensures that all the properties of Rx1 and Rx2 coincide in terms of interference, so it is sufficient to analyze Rx1. The channel output for Rx1 at $n^\mathrm{th}$ time slot becomes $y_{\rxi{1}}[n]=Q_{1}\mathcal{S}_{11}[0]x_{1}[n]+{I}_{1}[n]+{n}_{1}[n]$ from \eqref{eq_channel_output}. 
We approximate ${I}_{1}[n]\!=\!{ISI}_{1}[n]\!+\!{ILI}_{1}[n]$, given in \eqref{interferences}, with a Gaussian distribution that has mean $\mu_{I}$ and variance $\sigma_{I}^{2}$. Lemma~\ref{lemma_interference} provides the formulations for estimating the mean and the variance of the interference.
\begin{lemma}
\label{lemma_interference}
The $k^{\mathrm{th}}$ ISI term $Q_{x_{1}[n-k]}\mathcal{S}_{11}[k]$ in the summation of ({\ref{interferences}}) has mean value of $\pi_{1}Q_{1}A_{k}$ and variance of ${\pi_{1}Q_{1}A_{k}(1-A_{k})+\pi_{0}\pi_{1}Q_{1}^{2}A_{k}^{2}}$.
\end{lemma}
\begin{proof}
With probability $\pi_1$ and $\pi_0$, $Q_{x_{1}[n-k]}\mathcal{S}_{11}[k]$ follows $\mathcal{N}(Q_{1}A_{k},Q_{1}A_{k}(1-A_{k}))$ and becomes just zero, respectively. Therefore, the mean of the received ISI becomes $\pi_{1}Q_{1}A_{k}$ and the variance becomes
\begin{eqnarray}
\nonumber
\sigma^{2}&=&E[{x}^{2}]-{E[x]}^{2}\\ \nonumber
&=&\pi_{1}(Q_{1}^{2}A_{k}^{2}+Q_{1}^{2}A_{k}^{2}(1-A_{k})^{2})-\pi_{1}^{2}Q_{1}^{2}A_{k}^{2}\\ \nonumber
&=&\pi_{1}Q_{1}^{2}A_{k}^{2}(1-A_{k})^{2}+(\pi_{1}-\pi_{1}^{2})Q_{1}^{2}A_{k}^{2}\\ \nonumber
&=&\pi_{1}Q_{1}A_{k}(1-A_{k})+\pi_{0}\pi_{1}Q_{1}^{2}A_{k}^{2}.
\end{eqnarray}
\end{proof}

In a similar way, we can apply the lemma to find the mean and variance of the $k^{\mathrm{th}}$ ILI and sum both to find the total mean and variance of ${I}_{1}[n]$. The total interference mean and variance at the $n^{th}$ symbol slot becomes
\begin{align}
\label{eqn_interference_all}
\begin{split}
\mu_{I}&=\pi_{1}Q_{1}\left(\sum\limits_{k=1}^{n-1}A_{k}+\sum\limits_{k=0}^{n-1}B_{k}\right)\\
\sigma_{I}^{2}&=\pi_{0}\pi_{1}Q_{1}^{2}\left(\sum\limits_{k=1}^{n-1}A_{k}^2+\sum\limits_{k=0}^{n-1}B_{k}^2\right)\\
&+\pi_{1}Q_{1}\left(\sum\limits_{k=1}^{n-1}A_{k}(1-A_{k})+\sum\limits_{k=0}^{n-1}B_{k}(1-B_{k})\right)
\end{split}
\end{align}
where $B_{k}$ denotes the success probability of both $\mathcal{S}_{12}[k]$ and $\mathcal{S}_{21}[k]$.
Note that \eqref{eqn_interference_all} does not require the previously sent bit sequences. It  requires only the index of the current symbol, as it evaluates the expected value over the cases. Hence, \eqref{eqn_interference_all} can be used for each symbol consecutively.

After formulating the interference and the detector output, we can now derive the thresholds for the \textit{practical zero forcing}. We denote the probability density function of $\hat{\pmb{y}}_{p}$ when the transmitted bit is $0$ and $1$ as $\hat{\pmb{y}}_{p|0}$ and $\hat{\pmb{y}}_{p|1}$, respectively. We approximate the detector outputs at Rx$i$ by Gaussian distribution as
\begin{align}
\label{eq_detector_output}
\begin{split}
\hat{\pmb{y}}_{p|0}(i) &\sim\mathcal{N}\left(\mu_0,\sigma_0^2\right)=\mathcal{N}\left(\frac{{\mu_{I}}}{Q_{1}A_{0}},\frac{{\sigma_{I}^2}+{\sigma_{n}^2}}{Q_{1}^{2}A_{0}^2}\right)\\
\hat{\pmb{y}}_{p|1}(i) &\sim\mathcal{N}\left(\mu_1,\sigma_1^2\right)=\mathcal{N}\left(1\!+\!\mu_0,\frac{(1-A_0)}{Q_{1}A_0}\!+\!\sigma_0^2\right)
\end{split}
\end{align}
where $\hat{\pmb{y}}_{p|0}(i)$ is $i^{\mathrm{th}}$ element of the $2\times{1}$ vector. Note that $\hat{\pmb{y}}_{p|1}$ is the sum of $\hat{\pmb{y}}_{p|0}$ and ${\bar{\pmb{H}}}^{-1}\pmb{H}$ that determines the mean and the variance of $\hat{\pmb{y}}_{p|1}(i)$ in \eqref{eq_detector_output}. The formulations are obtained by utilizing \eqref{eqn_practical} and \eqref{eqn_b2n}.

Now we define a decision rule as $\arg\max(\hat{\pmb{y}}_{p|i})$ and need to find intersection points of two distributions (i.e., to find the decision threshold, $\eta_p$, for  \textit{practical zero forcing}). It leads to the equality 
\begin{equation}
\label{eq_decision_equation}
\frac{1}{\sigma_{0}\!\sqrt{2\pi}}\mathrm{exp}\left(\!-\frac{(\eta_p\!-\!\mu_0)^2}{2\sigma_0^2}\!\right)\!=\!\frac{1}{\sigma_{1}\!\sqrt{2\pi}}\mathrm{exp}\left(\!-\frac{(\eta_p\!-\!\mu_1)^2}{2\sigma_1^2}\!\right)
\end{equation}
and its solution in terms of $\eta_p$ becomes
\begin{equation}
\eta_p=\mu_{0}+\frac{-1\pm\sqrt{1+(\beta-1)(1+\sigma_0^{2}\beta\mathrm{ln}\beta)}}{\beta-1} \nonumber
\end{equation}
for $\beta=\left({\sigma_1}/{\sigma_0}\right)^2>1$. We denote the bigger one as $\eta_p^+$ and the smaller one as $\eta_p^-$, then the decision rule for the decoded bits $\hat{\pmb{x}}$ becomes:
\begin{equation}
\hat{\pmb{x}} = \delta(\hat{\pmb{y}}_{p}) = 
\left\{ 
	\begin{array}{ll}
	0 & \eta_p^{+}>\hat{\pmb{y}_p}>\eta_p^{-} \\
	1 & \text{otherwise}
	\end{array}
\right. \nonumber
\end{equation}
where $\delta(.)$ is the decision function at the receiver.

Note that $\beta \geq 1$ because $\sigma_1^2$ is the sum of $\sigma_0^2$ and the variance of ${\bar{\pmb{H}}}^{-1}\pmb{H}$. The case where $\beta$ becomes $1$ means that $\hat{\pmb{y}}_{p|0}$ and $\hat{\pmb{y}}_{p|1}$ have the same variances and it is trivial that the threshold becomes $\frac{\mu_{0}+\mu_{1}}{2}$.

We can obtain similarly the decision rule for \textit{adaptive thresholding}. We define $\hat{\pmb{y}}_{a|0}$ and $\hat{\pmb{y}}_{a|1}$ as in \eqref{eq_detector_output} and their means and variances, by Theorem~\ref{theorem_adap&prac}, become $A_{0}\mu_0$, $(A_{0}\sigma_{0})^2$, $A_{0}\mu_1$, and $(A_{0}\sigma_{1})^2$, respectively. Therefore, we have an equation similar to \eqref{eq_decision_equation} for finding the decision threshold of \textit{adaptive thresholding} method
\begin{equation}
\frac{A_{0}\sigma_{1}}{A_{0}\sigma_{0}}\,\mathrm{exp}\left(-\frac{(\eta_w-A_{0}\mu_0)^2}{2(A_{0}\sigma_0)^2}\right)=\mathrm{exp}\left(-\frac{(\eta_w-A_{0}\mu_1)^2}{2(A_{0}\sigma_1)^2}\right) \nonumber
\end{equation}
that can be solved similarly. 

\section{Results}
\label{results}
The system parameters used in this paper are given in Table~\ref{tab_result_params}.
We first give the definition of signal-to-interference-ratio (SIR) metric in molecular MIMO system and analyze the effect of  distance, $r_r$, and $h$. Next, we use the BER as the performance metric and analyze the effect of $Q_1$ and $t_s$. 
\begin{table}[t]
	\caption{Range of Parameters Used in the Analysis}
	\label{tab_result_params}
	\centering
	\begin{tabular}{L{3cm} L{1cm} L{3cm}}
	\hline
		Parameter &  Variable &Values \\
	\hline
		Diffusion cefficient  	& $D$  &      $50\, \mu m^2/s$ \\
		Distance 				& $d$  &     $\{ 2, 4 \} \mu m$ \\
		Radius of the receiver 	& $r_r$&      $\{ 2, 4 \} \mu m$ \\
		Bulge separation 		& $h$  &      $\{ 1, 2 \} \mu m$ \\
		\# molecules for sending bit-1 	&$Q_1$ &    $\{100\sim600\}$ molecules \\
		Probability of sending bit-1 & $\pi_1$ & 0.5 \\
		Symbol duration 		&$t_s$ &    $\{ 0.05\sim1 \}$ sec \\
		Molecular noise variance& $\sigma_{n}^{2}$& 10 \\
		Bit sequence length 	& 	   &    $5\times10^{4}$\\
		Replication 			&      & 20 \\
	\hline
	\end{tabular}
\end{table}

\subsection{SIR Analysis}
\label{result_SIR}
SIR is defined as the ratio of the expected number of molecules coming from the intended transmitter in the intended time slot and mean ILI plus ISI for just a one-shot signal. 
\begin{equation}
\text{SIR}=\frac{F_{11}(0,t_s)}{F_{11}(t_{s},\infty)+F_{12}(0,\infty)}. \nonumber
\end{equation}
Note that this definition is specific to the molecular communication case and explains the clearness of the mean signal term in the received signal. 
\begin{figure}[t]
	\centering
	\includegraphics[width=0.98\columnwidth,keepaspectratio]
	{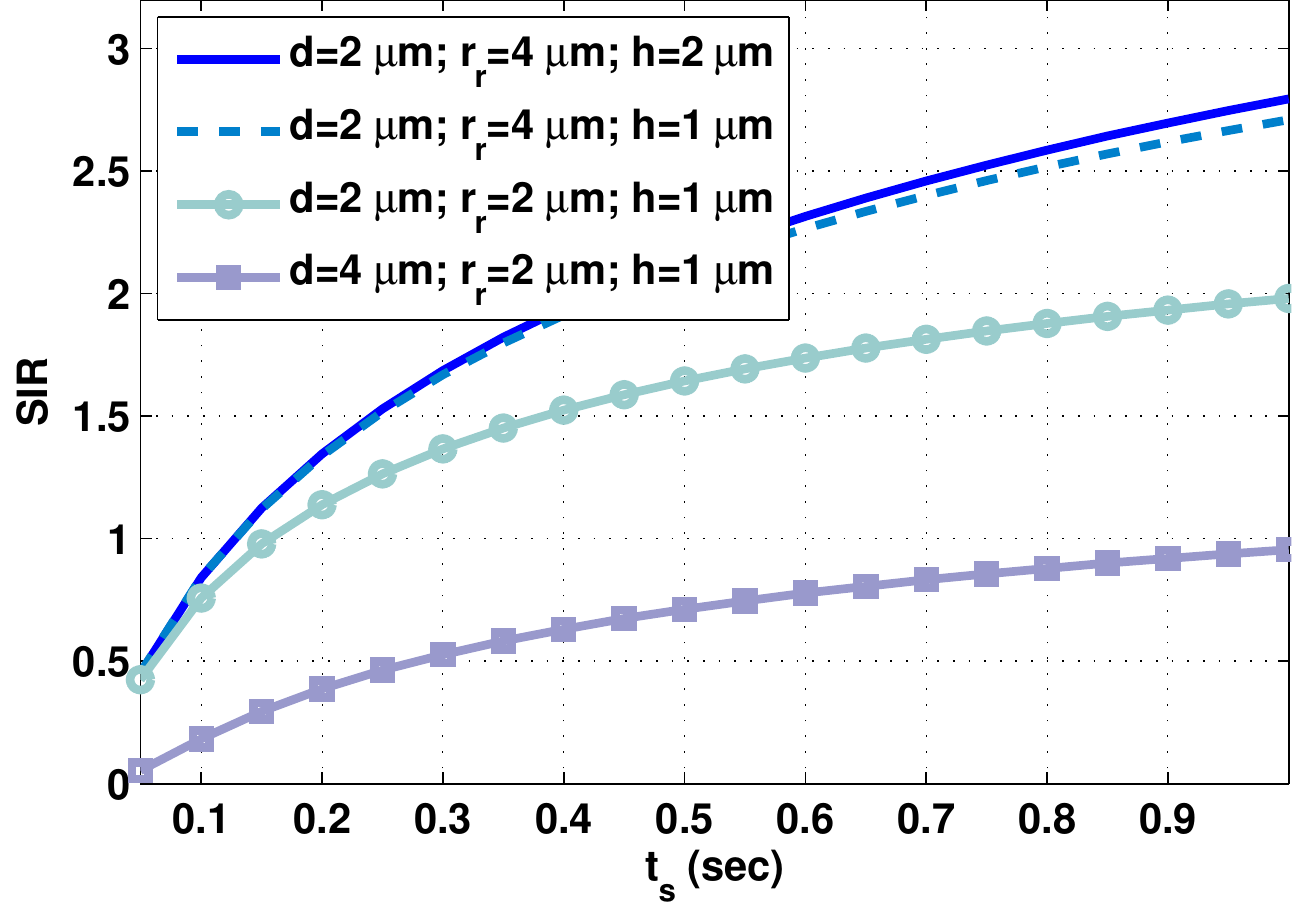}
	\caption{SIR plots of different topologies ($D=50 \,\mu m^2/s$)}
	\label{Fig:SIR}
\end{figure}

Fig.~\ref{Fig:SIR} depicts $t_s$ versus SIR for different topology parameters. With respect to SIR, the best enhancement is provided by reducing the distance between the transmitter and the receiver. Increasing the bulge size also gives merit, however, using more separation distance between Rx bulges results in only insignificant improvement. Thus, we can conclude that SIR is influenced more by the ISI term than the ILI term for our system setup.

We set the topological parameters $d=2{\mu}m$, $r_{r}=4{\mu}m$, $h=2{\mu}m$, and $D=50{\mu}{m}^{2}/s$ for the rest of the performance evaluation. The selected system parameters and the fitted values for model parameters are given in Table~\ref{tab_fitting_selected_params_all}. Utilizing fitted values enables us to estimate $F_{ij}(t)$ analytically.

\begin{table}[t]
	\caption{Fitted model parameters for the selected topology ($d=2\mu m$, $\rrn=4\mu m$, $h=2\mu m$, $D=50 \,\mu m^2/s$).}
	\label{tab_fitting_selected_params_all}
	\centering
	\begin{tabular}{l L{1.7cm} L{1.7cm} L{1.7cm}}
		\hline  
		Function                & $b_1$      & $b_2$        & $b_3$     \\
		\hline  
		$F_{11}(t)$          & 0.9155	     & 0.5236       & 0.5476	\\
		$F_{12}(t)$          & 0.1534	     & 0.2780       & 0.5363	\\
		\hline
	\end{tabular}
\end{table}

\subsection{BER Analysis}
\label{result_BER}
In this section, we analyze the BER with respect to varying $Q_1$ and $t_s$ and compare the performance gain of each of the proposed detection algorithms. Each Tx sends $5\times10^{4}$ bits with equal probability of sending bit-1 and bit-0. Most prior work has shown that, with an appropriate symbol duration, the current symbol is mostly affected by one previous symbol with the rest being negligible~\cite{kuran2010energyMF, kuran2012interferenceEO, kim2013novelMT}. Therefore, we consider four slots of interference in the simulations. We examined all the thresholds between $-1$ and $2$ with $10^{-3}$ interval and checked the optimal threshold for each $Q_1$. Empirically found fixed threshold, $\eta_f$, is selected as $0.2$. 

Fig.~\ref{Fig:BER_vs_Q} shows the BER performance of detection algorithms while $Q_1$ varies from 100 to 600. The first observation is that the \textit{adaptive thresholding} and the \textit{practical zero forcing} provide coinciding results, as in Theorem~\ref{theorem_adap&prac}. Increasing $Q_1$ (i.e., the signal power) decreases the BER for all the detection algorithms. However, the improvement of \textit{fixed thresholding} is significantly lower than the other methods. When the instantaneous $\pmb{H}$ is known, \textit{Genie-aided zero forcing} is applicable and gives the best performance. Obtaining that information, however is not easy. On the other hand, obtaining the optimal threshold by knowing $Q_1$ and $\pi_1$ is feasible and leads to a performance that is close to \textit{Genie-aided zero forcing}.
\begin{figure}
	\centering
	\includegraphics[width=0.95\columnwidth,keepaspectratio]
	{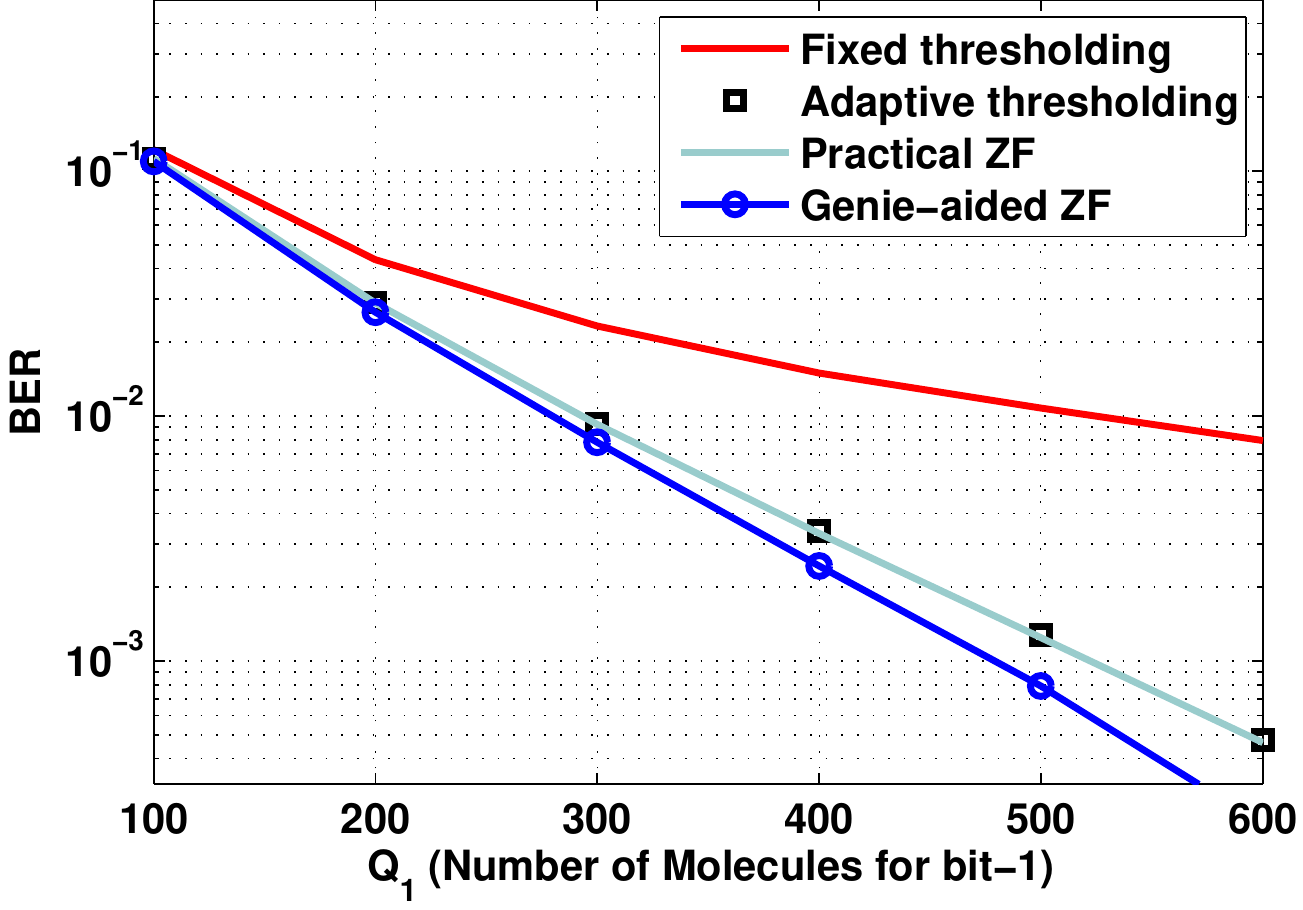}
	\caption{BER performance of detection algorithms ($t_s = 80\,ms$).}
	\label{Fig:BER_vs_Q}
\end{figure}

Fig.~\ref{Fig:BER_vs_ts} illustrates the BER performance against $t_s$ from $50\,ms$ to $130\,ms$. It shows that increasing $t_s$ gives faster improvement in terms of BER relative to increasing $Q_1$. This means that to achieve a lower BER it is more effective to decrease the information rate than to increase the transmit power.

\section{Conclusions}
\label{conclusion}
Data rates in molecular communications are acutely affected by interference. Therefore, any enhancement for molecular communications should consider interference effects precisely. In this paper, we proposed a MIMO system for the MCvD that takes into account inter symbol and inter link interference. Moreover, we proposed three symbol detection algorithms that depend on the information set that the receiver has. First, we modeled the channel's finite impulse response via fitting 3-D MIMO simulator results considering the fraction of the received molecules. We utilized the estimated function (that gives the fraction of received molecules) to determine interference and the optimal thresholds for the proposed methods. In the performance analysis, we first analyzed the effect of varying topological conditions on the SIR. The result shows that the transmitter-receiver distance and the size of receiver bulges (antennas) are more effective at reducing the interference rather than the separation of bulges. We investigated the performance of the proposed detection algorithms in terms of BER while varying $Q_1$ and $t_s$. We quantified the enhancement of these parameters in the molecular MIMO system. As a future direction, we will focus on a testbed implementation of the proposed algorithms while incorporating the drift.


\section*{Acknowledgment}
This research was in part supported by the MSIP (Ministry of Science, ICT \& Future Planning), Korea, under the ``IT Consilience Creative Program" (NIPA-2014-H0201-14-1002) supervised by the NIPA (National IT Industry Promotion Agency) and by the Basic Science Research Program (2014R1A1A1002186) funded by the Ministry of Science, ICT and Future Planning (MSIP), Korea, through the National Research Foundation of Korea.

\begin{figure}[t]
	\centering
	\includegraphics[width=0.95\columnwidth,keepaspectratio]
	{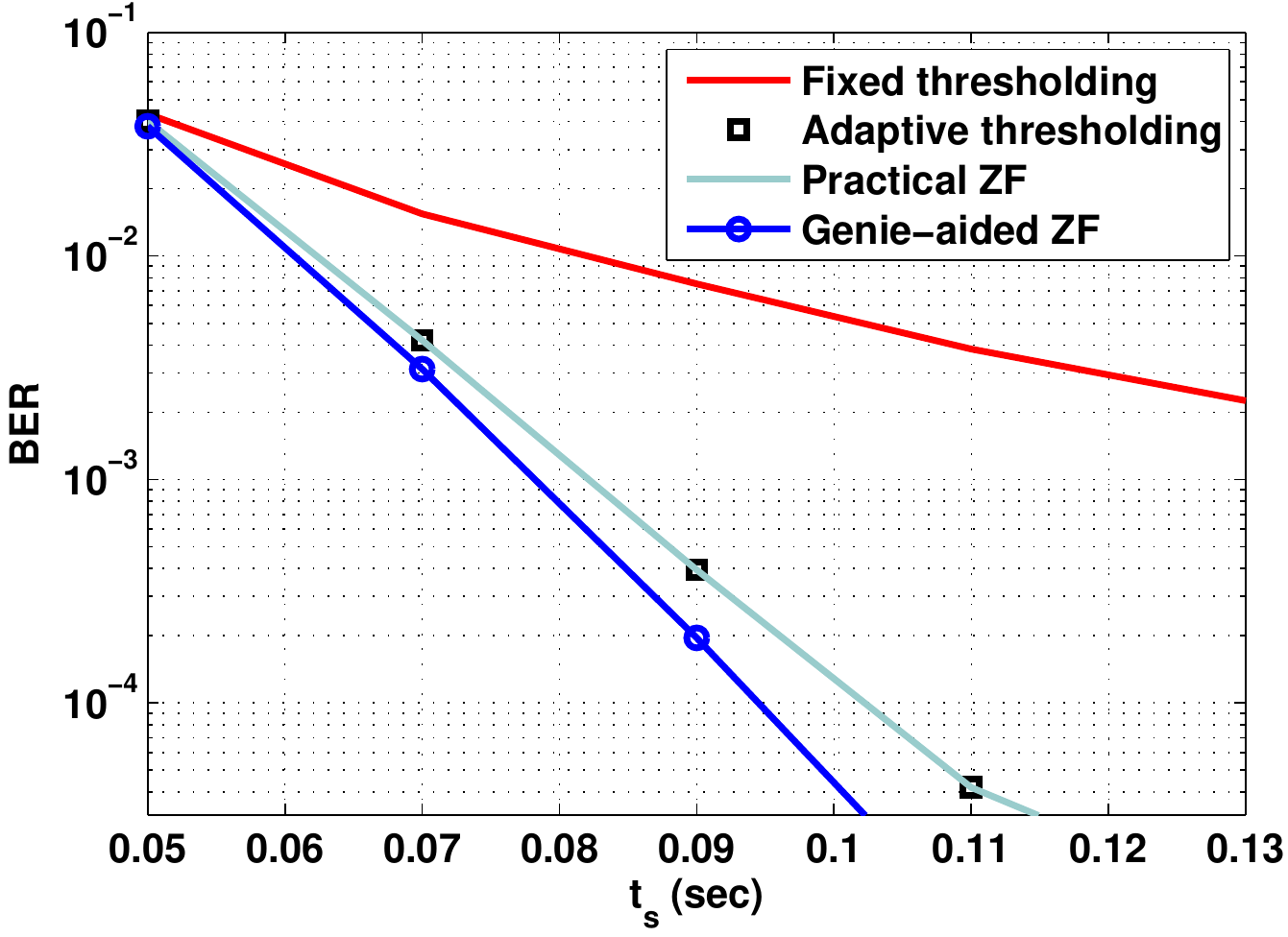}
	\caption{BER performance of detection algorithms ($Q_{1}=500$).}
	\label{Fig:BER_vs_ts}
\end{figure}

\bibliographystyle{IEEEtran}
\bibliography{references_molcom}

\end{document}